%% file: Risk_Measures_on_Lipschitz_Spaces.tex
\documentclass[a4paper,11pt,oneside]{article}

\input{preamble}

\begin{document}

\title{Risk Measures on Lipschitz Spaces}
\author{
Henrik Karlholm\thanks{Federal University of Rio Grande do Sul, Porto Alegre, Brazil}
\and
Marlon Moresco\footnotemark[1]
\and
Marcelo Righi\footnotemark[1]}
\date{
}
\maketitle

\begin{abstract}
This paper develops a theory of monetary risk measures on metric state spaces. We propose the space of Lipschitz functions vanishing at a reference state as a natural domain for financial positions. The associated Lipschitz-free space provides its canonical predual, linking anchored Lipschitz payoffs to transport-based dual variables interpreted as redistributions of mass around the benchmark. Since the domain lacks constants and need not be a Banach lattice under the Lipschitz norm, standard cash-additive methods do not apply directly. We address this by using additivity along benchmark-deviation instruments and derive dual representations for convex and coherent risk measures. The framework covers temporal cash flows, path-dependent payoffs, network risk, and model uncertainty.

\textbf{Keywords}: Risk measures; Lipschitz-free spaces; Robustness; Wasserstein metric; Polish spaces.
\end{abstract}

\section{Introduction}

The theory of monetary risk measures has, since the seminal contribution of
\citet{artzner1999}, been developed primarily on fixed probability spaces, with
classical Lebesgue spaces serving as the canonical domains. In this setting, the
dual representations of convex and coherent risk measures, established by
\citet{FoellmerSchied2002} and \citet{frittelli2002}, characterize capital
requirements through worst-case expectations over families of probability
measures, with penalty functions encoding the degree of model uncertainty. This
probabilistic framework has become fundamental in mathematical finance, both
because it provides an axiomatic account of capital requirements and because it
connects risk measurement to robust valuation and stress testing.

The choice of the underlying function space is, however, not merely technical.
It determines the admissible financial positions, the relevant continuity
properties, and the form of the dual variables. For this reason, a significant
part of the literature has extended risk measure theory beyond the classical
\(L^p\) and \(L^\infty\) domains. For instance, \citet{cheridito2009orlicz}
develop risk measures on Orlicz hearts in order to accommodate heavy-tailed
risks, while \citet{frittelli2014risk} take probability distributions themselves
as the primitive objects of evaluation. In a different direction,
\citet{cheridito2004} and \citet{cheridito2006} study risk measures on spaces of
stochastic processes, such as bounded c\`adl\`ag paths and discrete-time
cash-flow streams, motivated by path-dependent financial positions. The
question of the natural Banach space associated with a given risk functional is
also central in \citet{pichler2013}, who identifies Lorentz spaces in the case
of spectral risk measures. More recently, \citet{righi2025set} introduce set
risk measures, mapping closed and bounded sets of random variables into real
values, with applications to robust decision-making and systemic risk.

A particularly relevant contribution for the present paper is
\citet{delbaen2022cb}, who provides a robust characterization of monetary
utility functions on \(C_b(\Omega)\), the space of bounded continuous functions
on a Polish state space. This line of research shows that, once one moves away
from a fixed reference probability model, the geometry of the state space
becomes part of the risk measurement problem. Our contribution follows this
perspective, but replaces bounded continuous functions by benchmark-anchored
Lipschitz functions. We argue that, when the state space is metric and financial
positions are evaluated relative to a reference state, the natural domain is
the space of Lipschitz functions vanishing at that benchmark.

More precisely, let \((\Omega,d)\) be a Polish metric space and fix a benchmark
state \(\omega_0\in\Omega\). We take as primitive domain the Banach space
\(\Lip\) of real-valued Lipschitz functions satisfying \(X(\omega_0)=0\).
This space is not introduced as an arbitrary modelling restriction. Rather, it
is induced by the metric structure and by the benchmark itself. Each state
\(\omega\) determines the molecule
\(\delta_\omega-\delta_{\omega_0}\), which represents the deviation of
\(\omega\) from the benchmark at the level of point evaluations. The closed
linear span of these molecules is the Lipschitz-free space
\(\mathcal{F}(\Omega)\), whose dual is canonically identified with \(\Lip\).
Thus, benchmark-relative evaluation naturally leads to the Lipschitz-free
geometry of the state space; see also the perspective of
\citet{AliagaPernecka2021}, according to which Lipschitz functions vanishing at
a point play, in metric spaces, a role analogous to that of continuous
functions on compact Hausdorff spaces.

This geometry is closely connected to transport-based robustness. The
Wasserstein--\(1\) distance has become a standard way of quantifying model
uncertainty in risk measurement and distributionally robust optimization; see,
among others, \citet{kiesel2016}, \citet{blanchet2019quantifying}, and
\citet{glasserman2014robust}. By the Kantorovich--Rubinstein duality, the
natural test functions associated with Wasserstein--\(1\) transport are
precisely Lipschitz functions. Hence, once a benchmark state is fixed, the
anchored space \(\Lip\) provides the appropriate primal domain for risk
measures whose dual variables are described by transport-type redistributions
of mass. In this paper, we exploit this connection to obtain representations in
terms of signed measures with zero total mass and finite first moment. Such
measures naturally encode reallocations of mass relative to the benchmark.

Working on \(\Lip\) creates difficulties that are absent in the usual
Lebesgue-space setting. First, constants do not belong to \(\Lip\), except for
the zero function, since every position vanishes at the benchmark. Consequently,
classical cash additivity cannot be imposed. Second, although \(\Lip\) carries
the pointwise order and is a vector lattice, its Lipschitz norm need not be a
lattice norm: the implication \(|X|\leq |Y|\) need not imply
\(\|X\|_{\operatorname{Lip}}\leq\|Y\|_{\operatorname{Lip}}\). Thus, in
general, \(\Lip\) need not be a Banach lattice under the Lipschitz norm, and
standard arguments based on lattice norms or order-continuity are not directly
available. Third,
the relevant dual elements are not probability measures in their original
coordinates, but signed measures of zero total mass naturally paired with
benchmark-relative positions.

We address the absence of constants by replacing cash additivity with
additivity along a benchmark-deviation instrument. Following the general theory
of risk measures with eligible assets developed by
\citet{FarkasKochMunari2014} and \citet{Munari2016}, we fix an instrument
\(S\in\Lip\) and require invariance under shifts in the direction of \(S\).
In the present setting, however, \(S\) vanishes at \(\omega_0\) and therefore
cannot be bounded away from zero. It is not a numeraire in the usual sense, but
a benchmark-deviation instrument measuring the magnitude of departure from the
reference state.

The appropriate admissibility condition is that \(S\) be an order unit for the
pointwise order. Equivalently, \(S\) must dominate the benchmark distance:
there exists \(\kappa>0\) such that
$
S(\omega)\geq\kappa d(\omega,\omega_0).
$ 
The canonical choice is \(S=d(\cdot,\omega_0)\), but the admissible class also
contains other Lipschitz instruments that are uniformly comparable from below
with the benchmark distance. The order-unit property ensures that, for every
\(X\in\Lip\), there exists \(c_X>0\) such that \(|X|\leq c_XS\).
Consequently, the normalized payoff \(X/S\), extended by zero at
\(\omega_0\), is a bounded Borel function. The capital requirement can
therefore be interpreted as the number of units of the benchmark-deviation
instrument needed to restore acceptability, while \(X/S\) measures payoff per
unit of departure from the benchmark.

The resulting theory combines acceptance-set methods, \(S\)-additivity, and the
Lipschitz-free dual structure. For normalized, monotone, \(S\)-additive, convex,
and \(\sigma(\Lip,\Signed)\)-lower semicontinuous functionals, we obtain a
dual representation over positive normalized elements of \(\Signed\).
After a change of variables induced by \(S\), the same representation can be
expressed in terms of probability measures acting on the bounded normalized
payoff \(X/S\). In the coherent case, the penalty reduces to an indicator,
yielding a worst-case representation over a convex family of admissible
scenario measures.

The continuity requirements of the theory also reflect the geometry of the
domain. Under a mild boundedness condition on the sequence of capped
benchmark-deviation losses \(-\min(1,nS)\), the Fatou property for uniformly
Lipschitz-bounded pointwise convergence is equivalent to lower
semicontinuity with respect to both
\(\sigma(\Lip,\mathcal{F}(\Omega))\) and
\(\sigma(\Lip,\Signed)\). This equivalence connects the canonical
Lipschitz-free predual with the measure-representable dual pair used in the
risk-measure representation.

The framework covers several financial and economic environments in which the
state space has an intrinsic metric structure and a reference state.

\medskip
\noindent\textit{Temporal cash-flow profiles.}
Let \(T>0\) and let \(\Omega=[0,T]\) with its usual metric. A position is a
cumulative cash-flow profile \(X:[0,T]\to\R\) satisfying \(X(0)=0\) and a
Lipschitz bound. Choosing the benchmark \(\omega_0=0\) and the canonical
instrument \(S(t)=t\), the ratio \(X(t)/S(t)=X(t)/t\) represents profit per
unit of elapsed time. The Lipschitz condition implies that this temporal profit
intensity is uniformly bounded. Thus, the framework captures situations in
which gains and losses accumulate from an initial date at a controlled rate.

\medskip
\noindent\textit{Path-dependent financial positions.}
As in \citet{delbaen2022cb}, one may take \(\Omega\) to be a space of
trajectories, for instance \(C([0,T];\R^d)\) endowed with the supremum metric.
A benchmark trajectory \(\omega_0\) may represent a deterministic reference
path or baseline market evolution. Positions in \(\Lip\) are then
path-dependent payoffs whose values respond in a controlled manner to uniform
perturbations of the entire trajectory. The canonical instrument
\(S(\omega)=d(\omega,\omega_0)\) measures the maximal pathwise departure from
the reference trajectory, and \(X(\omega)/S(\omega)\) measures payoff per unit
of path deviation.

\medskip
\noindent\textit{Network and systemic risk.}
Motivated by \citet{hlavinova2024setvalued}, let \(\Omega\) be a countable set
of vertices representing financial institutions, equipped with the
shortest-path metric of a connected network. A benchmark vertex \(\omega_0\)
may represent a central bank, clearing counterparty, or reference institution.
With \(S=d(\cdot,\omega_0)\), the normalized payoff
\(X(\omega)/S(\omega)\) measures profit or loss per unit of graph distance from
the benchmark. This separates the magnitude of the evaluated position from the
institution's purely geometric location in the network.

\medskip
\noindent\textit{Model uncertainty and higher-order ambiguity.}
Let \(E\) be a Polish metric space and let \(\mathcal P_p(E)\) denote the set
of probability measures with finite \(p\)-th moment, endowed with the
Wasserstein--\(p\) metric. Since
\((\mathcal P_p(E),W_p)\) is itself Polish, it can serve as the state space.
Fixing a reference model \(\mu_0\), the canonical instrument
\(S(\mu)=W_p(\mu,\mu_0)\) measures the transport distance of a candidate model
from the benchmark. A position is then a Lipschitz functional on the space of
models, and \(X(\mu)/S(\mu)\) represents payoff per unit of model deviation.
This connects the framework with ambiguity-sensitive preferences and
higher-order uncertainty, as in \citet{gilboa_schmeidler_1989},
\citet{gilboa_schmeidler_1993}, \citet{epstein_schneider_2003}, and
\citet{ok2023lipschitz}.

These examples illustrate the common principle behind the paper. Once the state
space carries a metric and evaluation is anchored at a benchmark, monetary risk
measurement can be expressed through anchored Lipschitz positions,
order-unit benchmark instruments, and transport-based dual variables. The
Lipschitz-free approach therefore provides a functional-analytic foundation for
risk measurement in settings where robustness is inherently geometric rather
than tied to a fixed probability model.

The remainder of the paper is organized as follows. Section~2 introduces the
metric, measure-theoretic, and Lipschitz-free objects underlying the framework.
Section~3 defines the admissible benchmark instruments and develops the
correspondence between risk measures and their acceptance sets. In particular,
Propositions~\ref{Proposition Arho} and~\ref{Proposition A to rho} characterize the properties
transferred between a risk measure and its acceptance set and show how the
order-unit condition guarantees that the induced capital requirement is
finite-valued. Section~4 contains the main duality and continuity results.
\Cref{thm:convex-dual} establishes the signed-measure and probability
representations for convex risk measures, while Corollary~\ref{cor:coherent} gives the
corresponding coherent representation. \Cref{teo conti} establishes the
equivalence between the Fatou property and lower semicontinuity with respect to
the Lipschitz-free and measure-induced weak topologies. Finally, Section~5
illustrates the framework through worst-case, optimized-certainty-equivalent,
and Wasserstein-robust risk measures.

\section{Preliminaries}

We begin by introducing the mathematical objects that underlie our framework and their respective notation. Let \((\Omega,d)\) be a \emph{Polish metric space}, that is, a complete and separable metric space. We interpret \(\Omega\) as the \emph{state space} of the world, where each \(\omega\in\Omega\) denotes a possible state. The Borel \(\sigma\)-algebra generated by the metric topology is denoted by \(\mathcal B(\Omega)\).

Throughout, we adopt the following notation. 
We let \(\mathcal P(\Omega)\) denote the set of all Borel probability measures on \(\Omega\), and 
\(\mathcal P_1 (\Omega)\) the subset of measures with finite first moment, that is, \(\int_{\Omega} d(\omega,\omega_0)\, d\mu(\omega) < \infty\) for some (hence all) \(\omega_0\in\Omega\). 
We denote by \(\mathcal M(\Omega)\) the space of finite signed Borel measures on \(\Omega\), and by 
\(\Signed  (\Omega)\) the subspace of measures \(\eta\) such that \(\eta(\Omega)=0\) and 
\(\int_{\Omega} d(\omega, \omega_0)\,|\eta|(d\omega) < \infty\); in other words, measures of \emph{zero total mass} and \emph{finite first moment}. 
Every signed measure $\eta \in \mathcal{M}(\Omega)$ admits a unique \emph{Hahn--Jordan decomposition} $\eta = \eta^+ - \eta^-$, where $\eta^+$ and $\eta^-$ are nonnegative mutually singular Borel measures on $\Omega$ 
satisfying $|\eta| = \eta^+ + \eta^-$.
If \(\eta \in \Signed \), then \(\eta^{+}(\Omega) = \eta^{-}(\Omega) < \infty\), 
and both have finite first moment. 
If $\eta\neq0$, then $\eta^+(\Omega)=\eta^-(\Omega)>0$, and we may define
\(\mu := \eta^{+} / \eta^{+}(\Omega)\) and
\(\nu := \eta^{-} / \eta^{-}(\Omega)\), so that
\(\mu,\nu\in\mathcal P_1(\Omega)\). Hence every nonzero element of
\(\Signed(\Omega)\) can be represented as
\(\eta=c(\mu-\nu)\), where \(c>0\) and
\(\mu,\nu\in\mathcal P_1(\Omega)\); the case \(\eta=0\) is treated separately.

The space of real-valued Lipschitz functions on \(\Omega\) is denoted by 
\[
\operatorname{Lip}(\Omega)
:=\{X:\Omega\to\mathbb R:\;
\|X\|_{\mathrm{Lip}}
:=\sup_{\omega \neq \bar{\omega}}\tfrac{|X(\omega)-X(\bar{\omega})|}{d(\omega, \bar{\omega})}<\infty\}.
\] The Lipschitz seminorm \(\|\cdot\|_{\mathrm{Lip}}\) vanishes on constant functions, and  it cannot distinguish functions that differ only by a constant shift, hence, 
\(\operatorname{Lip}(\Omega)\) is only a seminormed space. 
To obtain a genuine Banach space, we fix a basepoint \(\omega_0\in\Omega\) and define 
\[
\operatorname{Lip}_0:=\{X\in\operatorname{Lip}(\Omega): X(\omega_0)=0\}.
\]
Throughout, we assume that $\Omega\neq\{\omega_0\}$. We also denote by $\mathcal{F}(\Omega)$ the Lipschitz-free space over $\Omega$, which is the canonical predual of $\Lip$.  We omit the explicit dependence on the space \(\Omega\) to simplify notation, writing, for example, \(\Lip := \operatorname{Lip}_0 (\Omega),
\quad
\mathcal{M}_0^1 := \Signed  (\Omega).\) The fact that $\operatorname{Lip}_0$ is a Banach space under the Lipschitz norm is a standard result, see  \cite{weaver2018lipschitz} or \cite{GodefroyKalton2003}, for instance.

The space $\operatorname{Lip}_0(\Omega)$ carries a natural pointwise partial order: we write 
$X \geq Y$ if $X(\omega) \geq Y(\omega)$ for all $\omega \in \Omega$. Under this order, 
$\operatorname{Lip}_0(\Omega)$ is an ordered Banach space whose positive cone
\begin{equation*}
    \operatorname{Lip}_0^+(\Omega) \coloneqq \{ X \in \operatorname{Lip}_0(\Omega) : X(\omega) \geq 0 
    \text{ for all } \omega \in \Omega \}
\end{equation*}
is closed in the Lipschitz norm and generating, in the sense that every element of  $\operatorname{Lip}_0(\Omega)$ can be written as the difference of two non-negative elements.  Moreover, the lattice operations -- pointwise maximum and minimum -- preserve Lipschitz  regularity, so that $\operatorname{Lip}_0(\Omega)$ is a Banach space and a vector lattice under the pointwise order. Although $\operatorname{Lip}_0(\Omega)$ is a vector lattice under the
pointwise order, its Lipschitz norm need not be a lattice norm. Indeed,
for general metric spaces, the implication
\(
|X|\leq |Y|
\quad\Longrightarrow\quad
\|X\|_{\mathrm{Lip}}\leq\|Y\|_{\mathrm{Lip}}
\)
may fail. Thus, in general, $\operatorname{Lip}_0(\Omega)$ is not a
Banach lattice with respect to the Lipschitz norm. Likewise, the
Lipschitz norm need not be order-continuous.

We take $\Lip$ to be our primal space of financial positions. We use the convention of writing \(X, Y \in \operatorname{Lip}_0\) to denote functions of profit and loss defined on the state space \(\Omega\).
All subsequent results and constructions will therefore be stated in terms of \(X\) and \(Y\). Henceforth, all pairings with measures are understood as
\[
\langle X,\eta\rangle := \int_{\Omega} X\,d\eta,
\]
for $X \in \operatorname{Lip}_0$ and $\eta \in \Signed $. Even if $\eta$ is not a probability measure, we also use the convention
\[
\int_{\Omega} X\,d\eta = \mathbb E_\eta [X]
\]
A key ingredient underlying our work is Wasserstein distances, a central family of metrics in optimal transport theory. We refer to \citet{villani2009optimal} for a comprehensive account of optimal transport theory and its applications. Let $\mathcal{P}_p(\Omega)$ be the space of probability measures with finite p-th moment. The Wasserstein--$p$ distance for each $p\geq1$ between $\mu$ and $\nu$ is the function $W_p: \mathcal{P}_p(\Omega) \times \mathcal{P}_p(\Omega) \rightarrow \mathbb{R}$ defined as
\[
W_p(\mu, \nu) := \left( \inf_{\pi \in \Pi(\mu, \nu)} \int_{\Omega \times \Omega} d(\omega, \bar{\omega})^p\, d\pi(\omega, \bar{\omega}) \right)^{1/p},
\]
where $\Pi(\mu, \nu)$ is the set of all couplings of $\mu$ and $\nu$, i.e., Borel measures on $\Omega \times \Omega$ with marginals $\mu$ and $\nu$. This is the Kantorovich relaxation of the classical Monge transport problem, which seeks a measurable transport map $T: \Omega \to \Omega$ pushing forward $\mu$ to $\nu$ (i.e., $T_\# \mu= \nu$) and minimizing the total transport cost $\int d(\omega, T(\omega))^p\, d\mu(\omega)$. 

Of particular importance in our framework is the case $p = 1$, where the Wasserstein--1 distance admits a dual representation in terms of Lipschitz-1 functions (see Remark 6.4 cf. \cite{villani2009optimal} for details) as \[
W_1(\mu,\nu) \;=\; \sup_{\substack{\|X\|_{\mathrm{Lip}} \leq 1}} \left( \int_\Omega X \, d\mu\;-\; \int_\Omega X \, d\nu \right),
\] This is the so-called \textbf{Kantorovich--Rubinstein Duality}. Equip $\mathcal{M}_0^1(\Omega)$ with the Kantorovich--Rubinstein norm
\[
\|\eta\|_{\mathrm{KR}}
:=
\sup_{\|X\|_{\mathrm{Lip}}\leq 1}
\left|\int_\Omega X\,d\eta\right|.
\]
By the Kantorovich--Rubinstein duality, every
$\eta\in\mathcal{M}_0^1(\Omega)$ defines a continuous linear functional
on $\operatorname{Lip}_0(\Omega)$ through
$\langle X,\eta\rangle:=\int_\Omega X\,d\eta$, and the resulting map
\[
\mathcal{M}_0^1(\Omega)\longrightarrow\operatorname{Lip}_0(\Omega)^*
\]
is an isometric embedding. Thus, $\mathcal{M}_0^1(\Omega)$ provides a
natural measure-representable dual partner of
$\operatorname{Lip}_0(\Omega)$; see, for instance,
\citet{godefroy2015survey}.

\section{Risk measures and their acceptance sets}

We first define the class of admissible benchmark-deviation instruments. In the usual
theory of \(S\)-additive risk measures, the eligible asset is required to be compatible
with the order structure, typically by imposing strict positivity, or even positivity
bounded away from zero; see
\cite{FarkasKochMunari2014}. This condition allows the eligible asset to play the
role of a numeraire. In the present setting, however, every element of \(\Lip\) vanishes
at the benchmark \(\omega_0\), so no admissible instrument can be bounded away from
zero. The appropriate replacement is to require \(S\) to be an order unit for the
pointwise order. This means that every benchmark-anchored payoff can be dominated, up
to a scalar multiple, by \(S\). Thus \(S\) provides a benchmark-deviation direction along
which arbitrary positions in \(\Lip\) can be compensated.

\begin{definition}\label{eligibleassetclass}
We define the class of admissible eligible assets by
\[
\mathcal{S}
:=
\left\{
S\in\Lip^+ :
\text{ for every } X\in\Lip \text{ there exists } c_X>0
\text{ such that }  X(\omega)\le c_XS (\omega)\,,\forall \omega \in \Omega
\right\}.
\]
For \(X\in\Lip\) and \(S\in\mathcal S\), we define \(X/S\) on
\(\Omega\setminus\{\omega_0\}\) in the usual way and set
\[
\frac{X}{S}(\omega_0):=0
\]
by convention.
\end{definition}

\begin{remark}
If \(S\in\mathcal S\), then \(X/S\) is a bounded Borel function for every
\(X\in\Lip\). Indeed, by the order-unit property, there exists \(c_X>0\) such that
\[
        |X|\le c_XS,
\]
and therefore
\[
        \left|\frac{X(\omega)}{S(\omega)}\right|\le c_X,
        \qquad \omega\ne\omega_0.
\]
Thus \(X/S\) is a bounded Borel function and belongs to
\(L^\infty(P)\) for every probability measure \(P\). In general, however, \(X/S\notin\Lip\). For instance, if
\(X=S\), then \(X/S\equiv 1\) on \(\Omega\setminus\{\omega_0\}\), while
\[
        \frac{X}{S}(\omega_0)=0.
\]
If \(\omega_0\) is an accumulation point of \(\Omega\), this function is not continuous
at \(\omega_0\), and hence it is not Lipschitz. Therefore \(X/S\) should be understood
as a bounded benchmark-normalized payoff, not as an element of the primal space.
\end{remark}

\begin{remark}
The canonical example is \(S=d(\cdot,\omega_0)\). Indeed, for every \(X\in\Lip\),
\[
        |X(\omega)|
        =
        |X(\omega)-X(\omega_0)|
        \le
        \|X\|_{\Lip}\,d(\omega,\omega_0).
\]
Hence \(d(\cdot,\omega_0)\in\mathcal S\). More generally, the order-unit condition
amounts to requiring that \(S\) dominate the benchmark distance up to a scalar multiple.
Thus admissible eligible assets are precisely benchmark-deviation instruments that are
strong enough, in the pointwise order, to control all Lipschitz positions anchored at
\(\omega_0\).
That is, every $S\in\mathcal S$ satisfies
\[
S(\omega)\geq \kappa d(\omega,\omega_0), \qquad \omega\in\Omega,
\]
for some $\kappa>0$. Conversely, this inequality is sufficient for $S$ to be an order unit, since every $X\in\operatorname{Lip}_0(\Omega)$ satisfies
$|X(\omega)|\leq \|X\|_{\operatorname{Lip}}d(\omega,\omega_0)$. Thus, an admissible instrument must control the metric distance from the benchmark uniformly. The following examples illustrate this requirement.

\noindent\textit{ Time-value of money.}
Let $\Omega=[0,T]$, endowed with the usual metric, fix $\omega_0=T$, and let $r>0$. The instrument
$S(t)=e^{-rt}-e^{-rT}$ measures the present-value advantage of receiving one unit of currency at time $t$ rather than at the benchmark date $T$. By the mean value theorem,
$S(t)\geq r e^{-rT}(T-t)$ for every $t\in[0,T]$. Therefore, $S$ dominates the distance to the benchmark and belongs to $\mathcal S$.

\noindent\textit{ Time and monetary uncertainty.}
Let $\Omega=[0,T]\times[\underline m,\overline m]$, endowed with the metric
$d((t,m),(s,q))=|t-s|+|m-q|$, and fix $\omega_0=(T,m_T)$. Define
$f(t,m)=me^{-r(T-t)}$. For any $\kappa>0$, the instrument
\[
S(t,m)
=
\bigl|f(t,m)-f(T,m_T)\bigr|
+
\kappa d\bigl((t,m),(T,m_T)\bigr)
\]
belongs to $\operatorname{Lip}_0^+(\Omega)$ and satisfies
$S\geq\kappa d(\cdot,\omega_0)$, so it is an order unit. The first term measures deviations in discounted monetary value, while the second ensures that every geometric departure from the benchmark is detected.
\end{remark}


In this paper, we take the approach of defining our risk measure axiomatically, building on the classical framework of convex and coherent risk measures in the tradition of \citet{follmer2016stochastic}.

\begin{definition} \label{def:risk}
Fix $S\in\mathcal S$. A risk measure $\rho:\operatorname{Lip}_0 \to \mathbb{R} $ is a convex risk measure if it is $\sigma(\Lip,\Signed)$-lower semicontinuous and satisfies the following conditions for all $X,Y \in \Lip$ 
\begin{enumerate}[noitemsep, topsep=0pt]
    \item \textbf{Monotonicity:} if \(X \ge Y\) pointwise, then \(\rho(X) \le \rho(Y)\). \label{rhomonotone}
    \item \textbf{S-additivity:} \(\rho(X + mS) = \rho(X) - m\), for   all  \(m \in \mathbb{R}\). \label{Sadditive}
    \item \textbf{Convexity:}  \(\rho(\lambda X + (1-\lambda)Y) \le \lambda \rho(X) + (1-\lambda)\rho(Y)\), for  any \(\lambda \in [0,1]\). \label{rhoconvex}

     \end{enumerate}
\noindent
    We say $\rho$ is a coherent risk measure if it is a convex risk measure and further satisfies  
    \begin{enumerate}[resume, noitemsep, topsep=0pt]
    \item \textbf{Positive homogeneity:} \(\rho(\lambda X) = \lambda \rho(X)\)  for  any $\lambda \in \mathbb{R}$ such that \(\lambda \ge 0\). \label{rhoposh}
\end{enumerate}
We assume, without  loss of generality that $\rho$ is normalized, i.e. $\rho(0)=0$. 

Any convex risk measure $\rho$ induces an  acceptance set by
 \[
A_\rho := \{\, X \in \operatorname{Lip}_0 : \rho(X) \le 0 \,\}.
\]   
\end{definition}

 An \textbf{Acceptance set} $A \subseteq \Lip$ is a set of acceptable financial positions, the acceptance set $A_\rho$ encodes the positions that are deemed acceptable under $\rho$ and plays a central role in the analysis of risk measures.  We also define risk measures $\rho_A$ through acceptance sets.

\begin{definition}\label{rho A} 
Let $A \subseteq \operatorname{Lip}_0$ be an acceptance set and let $S \in \mathcal{S}$ be a benchmark instrument. The risk measure induced by $A$ is defined by
\[
\rho_A(X)
:=
\inf \left\{\, m \in \mathbb{R} \; : \; X + mS \in A \right\},
\qquad X \in \operatorname{Lip}_0 .
\]
\end{definition}


\begin{proposition}\label{Proposition Arho}
Let $\rho:\operatorname{Lip}_0\to\mathbb R$ be a convex risk measure, and fix $S \in \mathcal S$, then $A_\rho$ satisfies the following properties for all $X, Y \in \Lip$ :
\begin{enumerate}[nosep]
    \item If $Y \in A_\rho$ and $X \ge Y$, then $X \in A_\rho$.\label{monotone}
    \item  \label{inf finite}$\inf\{m \in \R: mS \in A_\rho\} = 0$.

    \item $A_\rho$ is convex.\label{convex} 
    \item If $\rho$ is coherent, then $A_\rho$ is a convex cone.\label{cone}
    \item $A_\rho$ is $\sigma\bigl(\Lip(\Omega), \mathcal{M}_0^1(\Omega)\bigr)$-closed.\label{lipclosed}
    \item   $\rho_{A_\rho} = \rho$.\label{rhoA}
\end{enumerate}
\end{proposition}

\begin{proof}

\cref{monotone}: Remember that by assumption, $\rho$ is {monotone}, then let $Y \in A_\rho$ and suppose $X \ge Y$. 
Since $Y \in A_\rho$, we have $\rho(Y) \le 0$. 
By monotonicity, $\rho(X) \le \rho(Y) \le 0$, which implies $X \in A_\rho$. 


\cref{inf finite}: Direct calculation yields $\inf\{m \in \R: mS \in A_\rho\} = \inf\{m \in \R: \rho(mS) \leq 0 \} =  \inf\{m \in \R: m \geq 0\} = 0$

\cref{convex}: $A_\rho$ is a sublevel set of a convex function, hence, it is a convex set.

\cref{cone} If $\rho$ is {positively homogeneous},
then for any $X\in A_\rho$ and $\alpha\ge 0$ we have \(\rho(\alpha X)=\alpha\rho(X)\le 0,\) hence, $\alpha X\in A_\rho$ and $A_\rho$ is a convex cone.

\cref{lipclosed}: As $\rho$ is $\sigma(\Lip,\Signed)$-lower semicontinuous, all its sublevel sets are $\sigma(\Lip,\Signed)$-closed.

\cref{rhoA}: Let \(A_\rho := \{\, X \in \operatorname{Lip}_0 : \rho(X) \le 0 \,\}\) 
and define \(\rho_{A_\rho}(X) := \inf\{\, m \in \mathbb R : X + mS \in A_\rho \,\}\).
We show that \(\rho_{A_\rho} = \rho\).

\smallskip
\noindent(i) \(\rho_{A_\rho}(X) \le \rho(X)\):
Set \(m := \rho(X)\). 
By $S$-additivity, \(\rho(X + mS) = \rho(X) - m = 0\), 
so \(X + mS \in A_\rho\). 
Since \(m\) is admissible in the infimum defining \(\rho_{A_\rho}(X)\), it follows that
\[
\rho_{A_\rho}(X) = \inf\{\, r : X + rS \in A_\rho \,\} \le m = \rho(X).
\]

\smallskip
\noindent(ii) $\rho(X) \le \rho_{A_\rho}(X)$:
Let $m > \rho_{A_\rho}(X)$. By the definition of infimum, there exists 
$\tilde m < m$ such that $X + \tilde m S \in A_\rho$, so $\rho(X + \tilde m S) \le 0$.
Again using $S$--additivity,
\[
\rho(X) 
= \rho(X + \tilde m S) + \tilde m
\le 0 + \tilde m
< m.
\]
Since this holds for all $m > \rho_{A_\rho}(X)$, we conclude $\rho(X) \le \rho_{A_\rho}(X)$.

\smallskip
Combining (i) and (ii) yields $\rho(X) = \rho_{A_\rho}(X)$ for all 
$X \in \operatorname{Lip}_0$.

\end{proof}

\begin{proposition}\label{Proposition A to rho}
Fix $S\in\mathcal S$. Let $A\subseteq\operatorname{Lip}_0$ be upward
closed, convex, and $\sigma(\Lip,\Signed)$-closed, with \(
\inf\{m\in\R:mS\in A\}=0.
\)
Let $\rho_A$ be defined as in Definition~\ref{rho A}. Then $\rho_A$ is
finite-valued and the following
hold for all $X,Y\in\Lip$ and $\lambda,m\in\R$:
\begin{enumerate}[nosep]
    \item Monotonicity: If \(X \ge Y\), then \(\rho_A(X) \le \rho_A(Y)\).
    \label{rhoAmonotone}

    \item S-additivity:
    \(\rho_A(X + mS) = \rho_A(X) - m\) for all \(m \in \mathbb{R}\).
    \label{rhoASadditive}

    \item Convexity:
    \(\rho_A(\lambda X + (1-\lambda)Y) \le \lambda \rho_A(X) + (1-\lambda)\rho_A(Y)\), for \(\lambda \in [0,1]\).
    \label{rhoAconvex}

    \item Coherence: If $A$ is a cone, then
    \(\rho_A(\lambda X) = \lambda \rho_A(X)\) for all \(\lambda \ge 0\).
    \label{rhoAcoherent}

    \item \(A = A_{\rho_A} := \{X \in \Lip: \rho_A(X) \le 0\}.\)
    \label{Aequivalent}

    \item $\rho_A$ is $\sigma(\Lip,\Signed)$-lower semicontinuous.
    \label{rhoA_lsc}
\end{enumerate}
\end{proposition}
\begin{proof}
We first prove that $\rho_A$ is finite-valued. The assumption
$\inf\{m\in\R:mS\in A\}=0$ implies that $A$ is nonempty and
$\rho_A(0)=0$. Fix $X\in\Lip$ and choose $Y\in A$. Since $S$ is an
order unit, there exists $\lambda>0$ such that $Y-X\leq\lambda S$.
Thus $Y\leq X+\lambda S$, and upward closedness of $A$ yields
$X+\lambda S\in A$. Hence $\rho_A(X)\leq\lambda<\infty$.

To prove the lower bound, choose $c>0$ such that $X\leq cS$. If
$X+mS\in A$, then $(c+m)S\geq X+mS$, and upward closedness yields
$(c+m)S\in A$. Since $\inf\{r\in\R:rS\in A\}=0$, we have
$c+m\geq0$. Therefore every admissible $m$ satisfies $m\geq-c$, so
$\rho_A(X)>-\infty$.

\cref{rhoAmonotone}: Let \(X,Y\in\operatorname{Lip}_0\) with \(X\geq Y\).
Since \(X \ge Y\) and \(A\) is upward closed, 
every \(m\) that makes \(Y + mS\) acceptable also makes \(X + mS\) acceptable.
That is,
\[
\{\, m \in \mathbb R : Y + mS \in A \,\}
\subseteq 
\{\, m \in \mathbb R : X + mS \in A \,\}.
\]
Applying the monotonicity of the infimum to these sets gives
\[\rho_A(X) 
= \inf\{\, m : X + mS \in A \,\} 
\le \inf\{\, m : Y + mS \in A \,\}
= \rho_A(Y).\]
Therefore, \(\rho_A\) is monotone.

\cref{rhoASadditive}: To show that $\rho_A$ is $S$-additive, we take any $m \in \mathbb R$, then direct calculation and a change of variable yield
\[
\rho_A(X + mS)
= \inf\{\, r \in \mathbb R : X + (m + r)S \in A \,\} = \inf\{\, r-m \in \mathbb R : X + rS \in A \,\} =  \rho_A(X) -m .
\]
Hence, $\rho_A$ satisfies the $S$-additivity property.

\cref{rhoAconvex}:
By assumption, $A$ is convex. Let $\lambda \in [0,1]$ and $X,Y \in \operatorname{Lip}_0$.
Define $B := \{b \in \R : Y + bS \in A\} $ and $C :=\{ c \in \R : X+cS \in A\}$, then for any $c \in C$ and $b \in B$,
convexity of $A$ implies 
\[
\lambda (X + cS) + (1-\lambda)(Y + bS)
= (\lambda X + (1-\lambda)Y)
  + (\lambda c + (1-\lambda)b)S \in A.
\]
Hence, for any $c \in C$ and $b \in B$,    
$\lambda c + (1-\lambda)b  \in \{ m \in \R :  \lambda X +(1-\lambda) Y + mS  \in A\,\}$ and consequently
$\lambda c + (1-\lambda)b \geq \inf\{\, m \in \R :  \lambda X +(1-\lambda) Y + mS  \in A\,\}  = \rho_A(\lambda X + (1-\lambda)Y)$. 
Now, taking the infimum over all such $c,b$ yields
\[ \rho_A(\lambda X + (1-\lambda)Y) \leq \inf_{c \in C, b \in B} \{\lambda c + (1-\lambda) b\} =  \lambda \inf_{c \in C} \{ c \} +  (1-\lambda) \inf_{ b \in B} \{ b\} = \lambda \rho_A (X) + (1-\lambda)\rho_A (Y).
\]
Thus, $\rho_A$ is convex.

\cref{rhoAcoherent}: If, in addition, $A$ is a convex cone, then
$\rho_A(0)=0$. For every $\lambda>0$,
\[
\rho_A(\lambda X)
= \inf\{\, m : \lambda X + mS \in A \,\}
= \lambda \inf\{\, r : X + rS \in A \,\}
= \lambda \rho_A(X).
\]
Thus, $\rho_A$ is positively homogeneous and therefore coherent.

\cref{Aequivalent} The inclusion \(A\subseteq A_{\rho_A}\)
is immediate. Conversely, let \(X\in A_{\rho_A}\). Then
\(\rho_A(X)\le 0\), and hence, for every \(n\in\mathbb N\), there exists
\(m_n<1/n\) such that \(X+m_nS\in A\). Since \(S\ge 0\), we have
\(X+\frac1nS\ge X+m_nS\). By upward closedness of $A$, it follows that
\(X+\frac1nS\in A\) for every \(n\in\mathbb{N}\). Since \(X+\frac1nS\to X\) in norm, hence
also in \(\sigma(\Lip,\Signed)\), and \(A\) is \(\sigma(\Lip,\Signed)\)-closed,
we get \(X\in A\). Thus \(A=A_{\rho_A}\).

\cref{rhoA_lsc}: It is enough to show that every sublevel set of \(\rho_A\) is
\(\sigma(\Lip,\Signed)\)-closed, first, fix \(c\in\mathbb R\). By \(S\)-additivity of \(\rho_A\), we have
\(\rho_A(X+cS)=\rho_A(X)-c\). Therefore
\[
\{X\in\Lip:\rho_A(X)\le c\}
=
\{X\in\Lip:\rho_A(X+cS)\le 0\}
=
\{X\in\Lip:X+cS\in A\}.
\]
The map \(X\mapsto X+cS\) is continuous for \(\sigma(\Lip,\Signed)\), because
for every \(\eta\in \Signed\),
\(\langle X+cS,\eta\rangle=\langle X,\eta\rangle+c\langle S,\eta\rangle\).
Hence the last set is the inverse image of the \(\sigma(\Lip,\Signed)\)-closed
set \(A\) under a \(\sigma(\Lip,\Signed)\)-continuous map. Thus it is
\(\sigma(\Lip,\Signed)\)-closed.
Therefore all sublevel sets of \(\rho_A\) are \(\sigma(\Lip,\Signed)\)-closed,
and so \(\rho_A\) is \(\sigma(\Lip,\Signed)\)-lower semicontinuous.

\end{proof}

\section{Dual representations}

{Recall that $S \in \mathcal{S}$ is an order unit. This allows us to use $S$ as a density to reparametrize probability 
measures: given $P\in \mathcal{P}(\Omega)$ with $P(\{\omega_0\}) = 0$, we may 
define a new measure $\mu_S$ by weighting $P$ by $1/S$. The following lemma makes 
this precise and shows that integrating $X$ against $\mu_S$ is equivalent to 
integrating the discounted payoff $X/S$ against $P$, with the integral always finite 
and controlled by the Lipschitz norm of $X$. }



\begin{lemma}\label{lem:canonical-discounting}
Let $S \in \mathcal{S}$ and \(P\in\mathcal P(\Omega)\) satisfy \(P(\{\omega_0\})=0\), and define the Borel measure \(\mu_S\) by
\[
\mu_S(A):=\int_A \frac{1}{S}\,dP,\qquad A\in\mathcal B(\Omega),
\]
where \(1/S\) is understood as \(0\) at \(\omega_0\). Then, for every \(X\in\operatorname{Lip}_0(\Omega)\), the function \(X/S\), extended by \(0\) at \(\omega_0\), is Borel measurable and bounded. In particular, \(X\) is \(\mu_S\)-integrable and
\[
\int_\Omega X\,d\mu_S=\int_\Omega \frac{X}{S}\,dP.
\]
Moreover, if \(\int_\Omega 1/S\,dP<\infty\), then \(\mu_S\) is a finite positive Borel measure.
\end{lemma}

\begin{proof}
Let \(X\in\operatorname{Lip}_0(\Omega)\). Since \(X(\omega_0)=0\), we have
\[
|X(\omega)|\le \|X\|_{\operatorname{Lip}}\,d(\omega,\omega_0),\qquad \omega\in\Omega.
\]
As $S$ is an order unit, for all $X \in \Lip$  there is an $s_X<\infty$ such that \(|X|\leq s_X S\), it follows that, for $\omega \neq \omega_0$
\[
\left|\frac{X(\omega)}{S(\omega)}\right|\le s_X.
\]
Since \(X/S\) is defined as \(0\) at \(\omega_0\), this gives \(\|X/S\|_\infty\le s_X<\infty\). Moreover, \(X/S\) is continuous on \(\Omega\setminus\{\omega_0\}\), and its extension by \(0\) at \(\omega_0\) is Borel measurable.

Hence \(X/S\) is \(P\)-integrable. By the definition of \(\mu_S\),
\[
\int_\Omega |X|\,d\mu_S=\int_\Omega \frac{|X|}{S}\,dP=\int_\Omega \left|\frac{X}{S}\right|\,dP<\infty,
\]
so \(X\) is \(\mu_S\)-integrable. Again by the definition of \(\mu_S\), we obtain
\[
\int_\Omega X\,d\mu_S=\int_\Omega X\frac{1}{S}\,dP=\int_\Omega \frac{X}{S}\,dP.
\]
Finally, \(\mu_S(\Omega)=\int_\Omega 1/S\,dP\). Therefore, if \(\int_\Omega 1/S\,dP<\infty\), then \(\mu_S\) is finite.
\end{proof}

\begin{theorem}
\label{thm:convex-dual}
Fix $S\in\mathcal S$. A map
$\rho:\operatorname{Lip}_0(\Omega)\to\mathbb{R}$ is a convex risk
measure if and only if there exists a proper convex penalty function
$\alpha:C\to[0,+\infty]$, with $\inf_C\alpha=0$, such that
\[
\rho(X)
=
\sup_{\eta\in C}
\bigl\{
\langle -X,\eta\rangle-\alpha(\eta)
\bigr\},
\]
where
\[
C
:=
\Bigl\{
\eta\in\mathcal{M}_0^1(\Omega):
\langle X,\eta\rangle\geq0
\text{ for every }X\in\operatorname{Lip}_0^+(\Omega),
\ \langle S,\eta\rangle=1
\Bigr\}.
\]
The minimal penalty is
\[
\alpha_{\min}(\eta)
:=
\sup_{X\in A_\rho}\langle -X,\eta\rangle,
\qquad \eta\in C.
\]
Equivalently, there exists a proper convex penalty function
$\alpha_S:\mathcal{P}_S\to[0,+\infty]$, with
$\inf_{\mathcal{P}_S}\alpha_S=0$, such that
\[
\rho(X)
=
\sup_{P\in\mathcal{P}_S}
\left\{
-\mathbb{E}_P\!\left[\frac{X}{S}\right]-\alpha_S(P)
\right\},
\]
where
\[
\mathcal{P}_S
=
\left\{
P\in\mathcal{P}(\Omega):
P(\{\omega_0\})=0,\ 
\int_{\Omega}\frac{1}{S}\,dP<\infty
\right\}.
\]
The minimal probability penalty is
\[
\alpha_{S,\min}(P)
=
\sup_{X\in A_\rho}
\left\{
-\mathbb E_P\!\left[\frac{X}{S}\right]
\right\}.
\]
\end{theorem}

\begin{proof}
($\Rightarrow$) Recall that, by definition, \(\rho\) is normalized, monotone, \(S\)-additive, convex and
\(\sigma(\Lip,\Signed)\)-lower semicontinuous. By the Fenchel--Moreau theorem
for the dual pair \((\Lip,\Signed)\),
\[
\rho(X)=\sup_{\eta\in \Signed}\{\langle -X,\eta\rangle-\rho^*(\eta)\},
\qquad
\rho^*(\eta):=\sup_{Y\in\Lip}\{\langle -Y,\eta\rangle-\rho(Y)\}.
\]
We next characterize the effective domain of $\rho^*$. Let $\eta\in\Signed (\Omega)$.
By $S$--additivity,
\[
\rho^*(\eta)\ge
\sup_{m\in\mathbb R}
\{\langle -Y-mS,\eta\rangle-\rho(Y+mS)\}
=
\langle -Y,\eta\rangle-\rho(Y)
+\sup_{m\in\mathbb R}m(1-\langle S,\eta\rangle).
\]
Therefore, \(\rho^*(\eta)<\infty\) implies \(\langle S,\eta\rangle=1\).



Now let \(Z\in\Lip^+\). If \(\langle Z,\eta\rangle<0\), then, since
\(\lambda Z\ge 0\), for any real $\lambda>0$, monotonicity and normalization give
\(\rho(\lambda Z)\le \rho(0)=0\). Thus
\[
\rho^*(\eta)\ge
\langle -\lambda Z,\eta\rangle-\rho(\lambda Z)
\ge
-\lambda\langle Z,\eta\rangle
\to+\infty.
\]
 Hence, if $\rho^*(\eta) < +\infty$, \(\langle Z,\eta\rangle\ge 0\) for every
\(Z\in\Lip^+\). The set $C$ is nonempty since for every $\omega\neq\omega_0$ we have
\(
\frac{\delta_\omega-\delta_{\omega_0}}{S(\omega)}\in C.
\) Thus every element in the effective domain of \(\rho^*\)
belongs to \(C\), and the Fenchel--Moreau representation reduces to
\[
\rho(X)=\sup_{\eta\in C}\{\langle -X,\eta\rangle-\rho^*(\eta)\}.
\]
We now show that, on \(C\), the conjugate coincides with the acceptance-set
penalty. Let \(\eta\in C\). If \(Y\in A_\rho\), then
\(\rho(Y)\le 0\), and hence,
\(\langle -Y,\eta\rangle\le \langle -Y,\eta\rangle-\rho(Y)\le \rho^*(\eta)\).
Therefore \(\sup_{Y\in A_\rho}\langle -Y,\eta\rangle\le \rho^*(\eta)\).

Conversely, for arbitrary \(Y\in\Lip\), set \(\widetilde Y=Y+\rho(Y)S\).
By \(S\)-additivity, \(\rho(\widetilde Y)=0\), so
\(\widetilde Y\in A_\rho\). Since \(\langle S,\eta\rangle=1\),
\[
\langle -Y,\eta\rangle-\rho(Y)
=
\langle -\widetilde Y,\eta\rangle
\le
\sup_{Z\in A_\rho}\langle -Z,\eta\rangle.
\]
Taking the supremum over \(Y\in\Lip\) gives
\(\rho^*(\eta)\le \sup_{Z\in A_\rho}\langle -Z,\eta\rangle\). Hence the minimal penalty satisfies
\(\alpha_{\min}(\eta)=\rho^*(\eta)\) on \(C\). Setting
\(\alpha:=\alpha_{\min}\), we have \(\alpha\ge 0\), and since
\(\rho(0)=0\), the representation gives \(\inf_C\alpha=0\).
Convexity of \(\alpha\) follows because it is the supremum of affine
functions of \(\eta\).

It remains to rewrite the probability representation. Let
\(\eta\in C\). Since \(\eta\) is positive on \(\Lip^+\), its negative part is
concentrated at \(\omega_0\). Indeed, if \(\eta^-(\{\omega_0\}^c)>0\), choose
a compact set \(K\subset\{\omega_0\}^c\) such that \(\eta^-(K)>0\) and
\(\eta^+(K)=0\). By regularity, choose \(\varepsilon>0\) such that
\(U:=\{\omega:d(\omega,K)<\varepsilon\}\) does not contain \(\omega_0\) and
\(\eta^+(U)<\eta^-(K)\). The function
\(\varphi(\omega):=(1-d(\omega,K)/\varepsilon)^+\) belongs to \(\Lip^+\),
vanishes at \(\omega_0\), equals \(1\) on \(K\), and is supported in \(U\).
Thus
\(\langle \varphi,\eta\rangle
\le \eta^+(U)-\eta^-(K)<0\), contradicting positivity on \(\Lip^+\).
Therefore \(\eta^- = c\delta_{\omega_0}\). Since \(\eta(\Omega)=0\), we have
\(c=\eta^+(\Omega)\), and hence,
\[
\eta=\eta^+-\eta^+(\Omega)\delta_{\omega_0}.
\]

For $\eta\in C$, define $P_\eta$ by $dP_\eta=S\,d\eta^+$. Since
$\eta^-=\eta^+(\Omega)\delta_{\omega_0}$ and
$\langle S,\eta\rangle=1$, we have $P_\eta\in\mathcal{P}_S$ and
\[
\int_\Omega \frac{1}{S}\,dP_\eta=\eta^+(\Omega)<\infty.
\] Conversely, for $P\in\mathcal{P}_S$, define
$d\mu_S=S^{-1}\,dP$, with $S^{-1}(\omega_0)=0$, and set
$\eta_P:=\mu_S-\mu_S(\Omega)\delta_{\omega_0}$. Since
$S\geq\kappa d(\cdot,\omega_0)$ for some $\kappa>0$,
\[
\int_\Omega d(\omega,\omega_0)\,d\mu_S
=
\int_\Omega\frac{d(\omega,\omega_0)}{S(\omega)}\,dP
\leq\frac{1}{\kappa}.
\]
Hence $\eta_P\in\mathcal{M}_0^1(\Omega)$. Moreover,
$\langle Z,\eta_P\rangle\geq0$ for every $Z\in\Lip^+$,
$\langle S,\eta_P\rangle=1$, and
\[
\langle X,\eta_P\rangle
=
\int_\Omega\frac{X}{S}\,dP,
\qquad X\in\Lip.
\]
Therefore $\eta_P\in C$, and the two constructions are inverse to each
other. Thus $C$ and $\mathcal{P}_S$ are in bijective correspondence.

Define \(\alpha_S(P):=\alpha(\eta_P)\). For
\(\alpha=\alpha_{\min}\), this gives the minimal probability penalty
\[
\alpha_{S,\min}(P)=
\sup_{X\in A_\rho}
\left\{-\int_\Omega \frac{X}{S}\,dP\right\}.
\]
Substituting \(\eta=\eta_P\) in the dual representation yields
\[
\rho(X)=
\sup_{P\in\mathcal{P}_S}
\left\{
-\int_\Omega \frac{X}{S}\,dP-\alpha_S(P)
\right\}.
\]

($\Leftarrow$)
Conversely, suppose that \(\rho\) is defined by the representation. Since
$S$ is an order unit, for every $X\in\Lip$ there exists $c_X>0$ such
that $|X|\leq c_XS$. Thus, for every $\eta\in C$,
\(
|\langle X,\eta\rangle|
\leq \langle |X|,\eta\rangle
\leq c_X\langle S,\eta\rangle
=c_X.
\)
Since $\alpha\geq0$, the representation gives $\rho(X)\leq c_X$.
Moreover, for every $\varepsilon>0$, there exists
$\eta_\varepsilon\in C$ such that
$\alpha(\eta_\varepsilon)<\varepsilon$, and therefore
\[
\rho(X)
\geq
-\langle X,\eta_\varepsilon\rangle-\alpha(\eta_\varepsilon)
\geq
-c_X-\varepsilon.
\]
Letting $\varepsilon\downarrow0$ yields $\rho(X)\geq-c_X$. Hence
$\rho$ is finite-valued.

Since $\rho$ is the supremum of affine
$\sigma(\Lip,\Signed)$-continuous functions, it is convex and
$\sigma(\Lip,\Signed)$-lower semicontinuous. Since every \(\eta\in C\) is
positive on \(\Lip^+\), the representation is monotone. Since
\(\langle S,\eta\rangle=1\) for every \(\eta\in C\), we have
\(\rho(X+mS)=\rho(X)-m\). Finally,
\[
\rho(0)=\sup_{\eta\in C}\{-\alpha(\eta)\}=-\inf_C\alpha=0.
\]
Thus \(\rho\) is is normalized, monotone, \(S\)-additive, convex and
\(\sigma(\Lip,\Signed)\)-lower semicontinuous, i.e. a convex risk measure.

\end{proof}

\begin{corollary}[Coherent case]
\label{cor:coherent}
Fix \(S\in\mathcal S\). A map \(\rho:\Lip\to\mathbb R\) is a coherent risk
measure if and only if there exists a nonempty convex set
\(\mathcal{P}_S^\rho\subseteq \mathcal{P}_S\) such that, for every \(X\in\Lip\),
\[
\rho(X)
=
\sup_{P\in\mathcal{P}_S^\rho}
\left\{
-\mathbb E_P\!\left[\frac{X}{S}\right]
\right\} , \qquad X\in\operatorname{Lip}_0(\Omega).
\]
\end{corollary}

\begin{proof}
Assume first that \(\rho\) is coherent. By \cref{thm:convex-dual},
\(\rho\) admits the representation
\[
\rho(X)
=
\sup_{P\in\mathcal{P}_S}
\left\{
-\mathbb E_P\!\left[\frac{X}{S}\right]
-\alpha_S(P)
\right\},
\]
where \(\alpha_S\) is the minimal penalty. Since \(\rho\) is coherent, its
acceptance set \(A_\rho\) is a cone. Hence, for every \(P\in\mathcal{P}_S\),
\[
\alpha_S(P)
=
\sup_{X\in A_\rho}
\left\{
-\mathbb E_P\!\left[\frac{X}{S}\right]
\right\}
\in\{0,+\infty\}.
\]

Define
\[
\mathcal{P}_S^\rho
:=
\{P\in\mathcal{P}_S:\alpha_S(P)=0\}.
\]
Since \(\alpha_S\ge 0\) and \(\inf_{\mathcal{P}_S}\alpha_S=0\), the set
\(\mathcal{P}_S^\rho\) is nonempty. Moreover, it is convex, because
\(\alpha_S\) is convex and
\(\mathcal{P}_S^\rho=\{P\in\mathcal{P}_S:\alpha_S(P)\le 0\}\).
Therefore the previous representation reduces to
\[
\rho(X)
=
\sup_{P\in\mathcal{P}_S^\rho}
\left\{
-\mathbb E_P\!\left[\frac{X}{S}\right]
\right\}.
\]

Conversely, suppose that such a nonempty convex set
\(\mathcal{P}_S^\rho\subseteq\mathcal{P}_S\) exists. Define
\[
\alpha_S(P):=
\begin{cases}
0, & P\in\mathcal{P}_S^\rho,\\
+\infty, & P\in\mathcal{P}_S\setminus\mathcal{P}_S^\rho.
\end{cases}
\]
Since \(\mathcal{P}_S^\rho\) is nonempty and convex, \(\alpha_S\) is a proper
convex penalty with \(\inf_{\mathcal{P}_S}\alpha_S=0\). Hence, by
\cref{thm:convex-dual}, \(\rho\) is a convex risk measure. Moreover,
as
\[
\rho(X)
=
\sup_{P\in\mathcal{P}_S^\rho}
\left\{
-\mathbb E_P\!\left[\frac{X}{S}\right]
\right\},
\]
\(\rho\) is the supremum of linear maps. Therefore \(\rho\) is positively
homogeneous and subadditive. Hence, \(\rho\) is coherent.
\end{proof}

The next theorem completes the continuity analysis by establishing the equivalence between the Fatou property and lower semicontinuity with respect to both $\sigma(\operatorname{Lip}_0(\Omega),\mathcal{F}(\Omega))$ and $\sigma(\operatorname{Lip}_0(\Omega),\Signed (\Omega))$. The additional condition involving $\min(1,nS)$ has a natural financial interpretation: it describes benchmark-deviation losses whose exposure increases near the reference state but is capped at one, so that the associated loss remains uniformly bounded in magnitude. As $n$ increases, these payoffs converge pointwise to ${1}_{\Omega\setminus\{\omega_0\}}$, although they are not uniformly bounded in the Lipschitz norm because their sensitivity near the benchmark becomes progressively steeper. Requiring the capital assigned to $-\min(1,nS)$ to remain uniformly bounded prevents the risk measure from assigning arbitrarily large capital solely to increasingly sharp, but bounded, deviations from the benchmark. Under this condition, the theorem shows that the three continuity requirements provide equivalent formulations of the robustness of the risk measure.

\begin{theorem}\label{teo conti} Fix $S \in \mathcal{S}$.
    Let $\rho: \Lip \rightarrow \R$ be a convex, monotone and S-additive functional such that $\sup_{n\in \mathbb{N}} \rho( -\min(1,n S)) < \infty$. Then the following are equivalent:

    \begin{enumerate}[nosep]
    \item\label{teo conti fatou} $\rho$ satisfies  the following Fatou property:
    For any sequence $(X_n)\subset \Lip$ with $\sup_{n}\|X_n\|_{\operatorname{Lip}}<\infty$ such that
    $X_n(\omega)\to X(\omega)$ for all $\omega \in \Omega$, one has  
    \(\rho(X)\le \liminf_{n\to\infty}\rho(X_n).\) 
        \item\label{teo conti Flsc}  $\rho$ is $\sigma(\operatorname{Lip}_0(\Omega), \mathcal{F}(\Omega))$-lower semi-continuous 
        
        \item\label{teo conti Mlsc}  $\rho$ is $\sigma(\operatorname{Lip}_0(\Omega), \Signed (\Omega))$-lower semi-continuous
        
    \end{enumerate}

\end{theorem}

\begin{proof}

    \cref{teo conti fatou} $\Rightarrow $ \cref{teo conti Flsc}. 
Since $\Omega$ is Polish, it is separable, and therefore $\mathcal{F}(\Omega)$ is
separable. Indeed, if $D\subseteq\Omega$ is a countable dense subset
containing $\omega_0$, then the linear span of
$\{\delta_\omega-\delta_{\omega_0}:\omega\in D\}$ is dense in
$\mathcal{F}(\Omega)$. Consequently, every norm-closed ball of
$\operatorname{Lip}_0(\Omega)=\mathcal{F}(\Omega)^*$ is compact and metrizable for
$\sigma(\operatorname{Lip}_0(\Omega),\mathcal{F}(\Omega))$.

By $S$-additivity, it is enough to prove that $A_\rho$ is
$\sigma(\operatorname{Lip}_0(\Omega),\mathcal{F}(\Omega))$-closed. Since $A_\rho$ is
convex, by the Krein--Šmulian theorem it suffices to show that its
intersection with every norm-closed ball is weak$^\ast$ closed (in the $\sigma(\operatorname{Lip}_0(\Omega),\mathcal{F}(\Omega))$ sense).
Fix $r>0$, and let $(X_n)$ be a sequence in
\[
A_\rho\cap\{Y\in\operatorname{Lip}_0(\Omega):\|Y\|_{\operatorname{Lip}}\le r\}
\]
such that
$X_n\to X$ in $\sigma(\operatorname{Lip}_0(\Omega),\mathcal{F}(\Omega))$.
Since the closed ball is weak$^\ast$ closed, we also have
$\|X\|_{\operatorname{Lip}}\le r$.
For every $\omega\in\Omega$,
$\delta_\omega-\delta_{\omega_0}\in \mathcal{F}(\Omega)$ and, since every element
of $\operatorname{Lip}_0(\Omega)$ vanishes at $\omega_0$,
\[
X_n(\omega)
=
\langle X_n,\delta_\omega-\delta_{\omega_0}\rangle
\longrightarrow
\langle X,\delta_\omega-\delta_{\omega_0}\rangle
=
X(\omega).
\]
Moreover, $\sup_n\|X_n\|_{\operatorname{Lip}}\le r$. Hence, by the Fatou
property,
\[
\rho(X)\le\liminf_{n\to\infty}\rho(X_n)\le0.
\]
Thus $X\in A_\rho$. Therefore, the intersection of $A_\rho$ with every
norm-closed ball is weak$^\ast$ sequentially closed. Since these balls
are weak$^\ast$ metrizable, these intersections are weak$^\ast$ closed.
The Krein--Šmulian theorem now implies that $A_\rho$ is
$\sigma(\operatorname{Lip}_0(\Omega),\mathcal{F}(\Omega))$-closed.

Finally, for every $c\in\mathbb R$, $S$-additivity gives
\[
\{X\in\operatorname{Lip}_0(\Omega):\rho(X)\le c\}=A_\rho-cS.
\]
Thus every sublevel set of $\rho$ is
$\sigma(\operatorname{Lip}_0(\Omega),\mathcal{F}(\Omega))$-closed, and consequently
$\rho$ is $\sigma(\operatorname{Lip}_0(\Omega),\mathcal{F}(\Omega))$-lower
semicontinuous.

\cref{teo conti Flsc} $\Rightarrow$ \cref{teo conti Mlsc}.
Suppose that $\rho$ is
$\sigma(\operatorname{Lip}_0(\Omega),\mathcal{F}(\Omega))$-lower semicontinuous. By the
Fenchel--Moreau theorem applied to the dual pair
$(\operatorname{Lip}_0(\Omega),\mathcal{F}(\Omega))$, and arguing as in
\cref{thm:convex-dual}, we obtain
\[
\rho(X)
=
\sup_{\eta\in C_F}
\bigl\{
-\langle X,\eta\rangle-\alpha(\eta)
\bigr\},
\]
where
\[
C_F
=
\left\{
\eta\in \mathcal{F}(\Omega):
\langle X,\eta\rangle\ge0
\text{ for every }X\in\operatorname{Lip}_0^+(\Omega),
\ \langle S,\eta\rangle=1
\right\},
\]
and $\alpha(\eta)=\sup_{X\in A_\rho}\langle-X,\eta\rangle$.

Let $\eta\in C_F$ satisfy $\alpha(\eta)<\infty$. Since $\eta$ is a
positive element of $\mathcal{F}(\Omega)$ and $\Omega$ is Polish,
\citet[Corollary~5.8]{aliaga-pernecka-integral} yields a positive almost
Radon measure $\mu$ on $\Omega$ such that
$\mu(\{\omega_0\})=0$ and
$\langle X,\eta\rangle=\int_\Omega X\,d\mu$ for every
$X\in\operatorname{Lip}_0(\Omega)$. In particular, $\mu$ is
$\sigma$-finite. Moreover,
$\int_\Omega S\,d\mu=\langle S,\eta\rangle=1$.

Set $Y_n=\min(1,nS)$. Since $S\in\operatorname{Lip}_0^+(\Omega)$, we have
$Y_n\in\operatorname{Lip}_0^+(\Omega)$, and $Y_n\uparrow1$ on
$\Omega\setminus\{\omega_0\}$. From the dual representation,
\[
\int_\Omega Y_n\,d\mu-\alpha(\eta)
=
\langle Y_n,\eta\rangle-\alpha(\eta)
\le \rho(-Y_n).
\]
Therefore,
\[
\int_\Omega Y_n\,d\mu
\le
\sup_{k\in\mathbb N}\rho(-Y_k)+\alpha(\eta)
<\infty.
\]
By the monotone convergence theorem,
$\mu(\Omega)=\lim_n\int_\Omega Y_n\,d\mu<\infty$.

Since $S$ is an order unit, applied to
$d(\cdot,\omega_0)\in\operatorname{Lip}_0^+(\Omega)$ there exists $c>0$ such
that $d(\omega,\omega_0)\le cS(\omega)$ for every $\omega\in\Omega$.
Consequently,
$\int_\Omega d(\omega,\omega_0)\,d\mu\le c\int_\Omega S\,d\mu=c$.

It follows that
$\mu-\mu(\Omega)\delta_{\omega_0}\in \Signed (\Omega)$. Since every
$X\in\operatorname{Lip}_0(\Omega)$ vanishes at $\omega_0$, this measure
represents the same functional as $\eta$:
\[
\int_\Omega X\,d\bigl(\mu-\mu(\Omega)\delta_{\omega_0}\bigr)
=
\int_\Omega X\,d\mu
=
\langle X,\eta\rangle.
\]
Hence every element of $C_F$ with finite penalty is represented by an
element of $\Signed (\Omega)$. Elements with infinite penalty do not
contribute to the supremum, so the dual representation of $\rho$ can be
reduced to $\Signed (\Omega)$.

Thus $\rho$ is the supremum of affine
$\sigma(\operatorname{Lip}_0(\Omega),\Signed (\Omega))$-continuous functions.
Therefore, $\rho$ is
$\sigma(\operatorname{Lip}_0(\Omega),\Signed (\Omega))$-lower semicontinuous.

\cref{teo conti Mlsc} $\Rightarrow$ \cref{teo conti fatou}.
Let $(X_n)\subset \operatorname{Lip}_0(\Omega)$ satisfy
$\sup_n\|X_n\|_{\operatorname{Lip}}<\infty$ and
$X_n(\omega)\to X(\omega)$ for every $\omega\in\Omega$.
Choose $R>0$ such that $\|X_n\|_{\operatorname{Lip}}\le R$ for every
$n$. Since $X_n(\omega_0)=0$, we have
$|X_n(\omega)|\le R\,d(\omega,\omega_0)$ for every $\omega\in\Omega$.
Moreover, $X(\omega_0)=0$ and, for every $\omega,\bar\omega\in\Omega$,
\[
|X(\omega)-X(\bar\omega)|
=
\lim_{n\to\infty}|X_n(\omega)-X_n(\bar\omega)|
\le R\,d(\omega,\bar\omega).
\]
Hence $X\in\operatorname{Lip}_0(\Omega)$ and
$|X(\omega)|\le R\,d(\omega,\omega_0)$.

Let $\eta\in \Signed (\Omega)$. Since
$|X_n(\omega)-X(\omega)|\le 2R\,d(\omega,\omega_0)$ and
$d(\cdot,\omega_0)$ is integrable with respect to $|\eta|$, the dominated
convergence theorem yields
$\int_\Omega |X_n-X|\,d|\eta|\to0$. Consequently,
$|\langle X_n-X,\eta\rangle|
\le\int_\Omega |X_n-X|\,d|\eta|\to0$.

Since this holds for every $\eta\in \Signed (\Omega)$, we conclude that
$X_n\to X$ in
$\sigma(\operatorname{Lip}_0(\Omega),\Signed (\Omega))$. Therefore, the
$\sigma(\operatorname{Lip}_0(\Omega),\Signed (\Omega))$-lower semicontinuity of
$\rho$ implies $\rho(X)\le\liminf_{n\to\infty}\rho(X_n).$
Thus $\rho$ satisfies the Fatou property.

\end{proof}

\begin{remark}
The condition  $\sup_{n\in \mathbb{N}} \rho( -\min(1,n S)) < \infty$ in the \cref{teo conti} is implied by monotonicity and S-additivity if $\omega_0$ is an isolated point of $\Omega$. 

If $\omega_0$ is isolated, there exists $a>0$ such that
$d(\omega,\omega_0)\ge a$ for every $\omega\ne\omega_0$. Since $S$ is
an order unit, there exists $c>0$ such that
$d(\omega,\omega_0)\le cS(\omega)$. Hence
$\min(1,nS)\le(c/a)S$ for every $n$. By monotonicity and
$S$-additivity,
$\rho\bigl(-\min(1,nS)\bigr)
\le
\rho\bigl(-(c/a)S\bigr)
=
\rho(0)+c/a.
$ Therefore,
$\sup_{n\in\mathbb N}\rho\bigl(-\min(1,nS)\bigr)<\infty$.
\end{remark}

\section{Examples}

\begin{example}[Canonical worst-case functional]

Let us illustrate our framework with examples of possible applications. We begin by presenting a simple but important primal object, a worst-case risk measure over the dual slice. 

Let $(\Omega,d)$ be a Polish metric space with basepoint $\omega_0$ and benchmark 
deviation instrument $S \in \mathcal{S}$. Define
\[
\mathcal{P}_S := \left\{ P \in \mathcal{P}(\Omega) : P(\{\omega_0\}) = 0, 
\int_\Omega \frac{1}{S} \, dP < \infty \right\}.
\]
For $X \in \operatorname{Lip}_0(\Omega)$, consider the worst-case benchmark-normalized 
functional
\[
\rho_0(X) := \sup_{P \in \mathcal{P}_S} \mathbb{E}_P\!\left[-\frac{X}{S}\right]
= \sup_{P \in \mathcal{P}_S} \int_\Omega -\frac{X(\omega)}{S(\omega)} \, dP(\omega).
\]
Writing $f(\omega) := -X(\omega)/S(\omega)$, we establish that $\rho_0(X) = 
\sup_{\omega \in \Omega\setminus \{\omega_0\}} f(\omega)$ via a two-sided argument. For the lower bound, 
since $S(\omega) > 0$ on $\Omega \setminus \{\omega_0\}$, the Dirac measure 
$\delta_\omega \in \mathcal{P}_S$ for every $\omega \neq \omega_0$, and
\[
\mathbb{E}_{\delta_\omega}\!\left[-\frac{X}{S}\right] = -\frac{X(\omega)}{S(\omega)} 
= f(\omega).
\]
Taking the supremum over $\omega$ yields $\rho_0(X) \geq \sup_{\omega} f(\omega)$. 
For the upper bound, since $f(\omega) \leq \sup_{\bar\omega} f(\bar\omega)$ pointwise, 
integrating against any $P \in \mathcal{P}_S$ gives
\[
\mathbb{E}_P\!\left[-\frac{X}{S}\right] = \int f \, dP \leq \int \sup_{\bar\omega} 
f(\bar\omega) \, dP = \sup_{\bar\omega} f(\bar\omega),
\]
where the last equality uses $P(\Omega) = 1$. Taking the supremum over $P \in 
\mathcal{P}_S$ yields $\rho_0(X) \leq \sup_{\omega} f(\omega)$. Combining both 
inequalities,
\[
\rho_0(X) = \sup_{\omega \in \Omega\setminus\{\omega_0\}} \left(-\frac{X(\omega)}{S(\omega)}\right) 
= -\inf_{\omega \in \Omega\setminus\{\omega_0\}}\frac{X(\omega)}{S(\omega)}.
\]
This is a coherent risk measure in the sense of Corollary~\ref{cor:coherent} with $\mathcal{P}_S^\rho 
= \mathcal{P}_S$ and vanishing penalty $\alpha_S \equiv 0$: the dual formulation 
collapses to pointwise worst-case evaluation of the benchmark-normalized payoff.

\end{example}

\begin{example}[Benchmark-adapted optimized certainty equivalent]
\label{ex:oce}

Let \(S\in\mathcal S\), and let \(g:\Lip\to\R\) satisfy the following conditions:
\begin{enumerate}[noitemsep,topsep=0pt]
    \item \(g\) is increasing with respect to the pointwise order;
    \item \(g\) is convex and \(\sigma(\Lip,\Signed)\)-lower semicontinuous;
    \item the function \(h:\R\to\R\), defined by
    \(h(y):=-y+g(yS)\), satisfies
    \(\inf_{y\in\R}h(y)=0\) and \(h(y)\to+\infty\) as \(|y|\to\infty\);
    \item \(\sup_{n\in\mathbb N}g(\min(1,nS))<\infty\).
\end{enumerate}

Define
\[
\rho_g(X):=\inf_{y\in\R}\{-y+g(yS-X)\},
\qquad X\in\Lip.
\]

We first show that \(\rho_g\) is finite-valued. Since \(S\) is an order unit,
for every \(X\in\Lip\) there exists \(c_X>0\) such that
\(|X|\leq c_XS\). Hence \(yS-X\geq(y-c_X)S\), and monotonicity of \(g\)
yields
\[
-y+g(yS-X)
\geq
-y+g((y-c_X)S)
=
h(y-c_X)-c_X
\geq-c_X.
\]
Taking \(y=0\) also gives \(\rho_g(X)\leq g(-X)<\infty\). Therefore,
\(-c_X\leq\rho_g(X)\leq g(-X)\).

The same inequality shows that
\(y\mapsto-y+g(yS-X)\) is coercive. It is also lower semicontinuous,
since \(y\mapsto yS-X\) is \(\sigma(\Lip,\Signed)\)-continuous.
Thus, the infimum defining \(\rho_g(X)\) is attained. Moreover, the
normalization of \(h\) gives \(\rho_g(0)=0\).

For \(m\in\R\), the change of variables \(z=y-m\) gives
\[
\rho_g(X+mS)
=
\inf_{z\in\R}\{-(z+m)+g(zS-X)\}
=
\rho_g(X)-m,
\]
so \(\rho_g\) is \(S\)-additive. If \(X\geq Y\), then
\(yS-X\leq yS-Y\), and monotonicity of \(g\) implies
\(\rho_g(X)\leq\rho_g(Y)\).

To prove convexity, let \(X,Y\in\Lip\), \(\lambda\in[0,1]\), and
\(y,z\in\R\). Using \(\lambda y+(1-\lambda)z\) as an admissible value
in the definition of \(\rho_g(\lambda X+(1-\lambda)Y)\), convexity of
\(g\) yields
\begin{align*}
\rho_g(\lambda X+(1-\lambda)Y)
&\leq
-\lambda y-(1-\lambda)z \\
&\quad
+g\bigl(\lambda(yS-X)+(1-\lambda)(zS-Y)\bigr)\\
&\leq
\lambda[-y+g(yS-X)]
+(1-\lambda)[-z+g(zS-Y)].
\end{align*}
Taking the infimum over \(y\) and \(z\) proves that \(\rho_g\) is convex.

We next verify the Fatou property. Let \((X_n)\subset\Lip\) satisfy
\(\sup_n\|X_n\|_{\operatorname{Lip}}\leq R\) and
\(X_n(\omega)\to X(\omega)\) for every \(\omega\in\Omega\).
Since \(S\in\mathcal S\), there exists \(\kappa>0\) such that
\(S\geq\kappa d(\cdot,\omega_0)\). Consequently,
\[
|X_n|\leq\frac{R}{\kappa}S
\qquad\text{and}\qquad
|X|\leq\frac{R}{\kappa}S.
\]
Set \(c:=R/\kappa\). The previous lower bound implies
\(\rho_g(X_n)\geq-c\) for every \(n\).

If \(\liminf_n\rho_g(X_n)=+\infty\), the conclusion is immediate.
Otherwise, pass to a subsequence, still denoted by \((X_n)\), such that
\(\rho_g(X_n)\to\liminf_k\rho_g(X_k)<\infty\), and let \(y_n\) be a
minimizer for \(\rho_g(X_n)\). Then
\[
h(y_n-c)-c
\leq
-y_n+g(y_nS-X_n)
=
\rho_g(X_n).
\]
Coercivity of \(h\) implies that \((y_n)\) is bounded. Passing to a
further subsequence, assume that \(y_n\to y\).

As observed in the proof of \cref{teo conti}, uniformly
Lipschitz-bounded pointwise convergence implies convergence in
\(\sigma(\Lip,\Signed)\). Hence
\(y_nS-X_n\to yS-X\) in \(\sigma(\Lip,\Signed)\). By the lower
semicontinuity of \(g\),
\[
\rho_g(X)
\leq
-y+g(yS-X)
\leq
\liminf_{n\to\infty}\{-y_n+g(y_nS-X_n)\}
=
\liminf_{n\to\infty}\rho_g(X_n).
\]
Thus, \(\rho_g\) satisfies the Fatou property.

Furthermore, taking \(y=0\) gives
\[
\rho_g(-\min(1,nS))
\leq
g(\min(1,nS)).
\]
Therefore,
\(\sup_n\rho_g(-\min(1,nS))<\infty\). By \cref{teo conti}, the Fatou
property is equivalent, under this condition, to
\(\sigma(\Lip,\Signed)\)-lower semicontinuity. Consequently,
\(\rho_g\) is a convex risk measure.

We now identify its dual penalty. Define the conjugate of \(g\), with
respect to the dual pair \((\Lip,\Signed)\), by
\[
g^*(\eta)
:=
\sup_{Z\in\Lip}\{\langle Z,\eta\rangle-g(Z)\},
\qquad \eta\in\Signed.
\]
The conjugate of \(\rho_g\), under the sign convention used for risk
measures, satisfies
\begin{align*}
\rho_g^*(\eta)
&=
\sup_{\substack{X\in\Lip\\y\in\R}}
\{\langle-X,\eta\rangle+y-g(yS-X)\}\\
&=
\sup_{\substack{Z\in\Lip\\y\in\R}}
\{\langle Z,\eta\rangle-g(Z)
+y(1-\langle S,\eta\rangle)\}\\
&=
\begin{cases}
g^*(\eta), & \langle S,\eta\rangle=1,\\
+\infty, & \langle S,\eta\rangle\neq1,
\end{cases}
\end{align*}
where the second equality follows from the substitution \(Z=yS-X\).

Monotonicity of \(g\) implies that \(g^*(\eta)<\infty\) only if
\(\langle Z,\eta\rangle\geq0\) for every \(Z\in\Lip^+\). Indeed, if
\(\langle Z,\eta\rangle<0\) for some \(Z\geq0\), then
\(g(-\lambda Z)\leq g(0)\), while
\[
\langle-\lambda Z,\eta\rangle-g(-\lambda Z)
\longrightarrow+\infty
\]
as \(\lambda\to\infty\).

It follows from \cref{thm:convex-dual} that
\[
\rho_g(X)
=
\sup_{\eta\in C}
\{-\langle X,\eta\rangle-g^*(\eta)\},
\]
where
\[
C
=
\left\{
\eta\in\Signed:
\langle Z,\eta\rangle\geq0
\text{ for every }Z\in\Lip^+,
\quad
\langle S,\eta\rangle=1
\right\}.
\]
Thus, the minimal penalty of \(\rho_g\) is the restriction of \(g^*\)
to \(C\). Equivalently, using the probability parametrization in
\cref{thm:convex-dual},
\[
\rho_g(X)
=
\sup_{P\in\mathcal{P}_S}
\left\{
-\mathbb E_P\left[\frac{X}{S}\right]
-g^*(\eta_P)
\right\},
\]
where
\[
\mu_S(A):=\int_A\frac{1}{S}\,dP,
\qquad
\eta_P:=\mu_S-\mu_S(\Omega)\delta_{\omega_0}.
\]

A concrete instance is obtained by fixing \(P_0\in\mathcal{P}_S\) such
that
\(\mathbb E_{P_0}[\exp(1/S)]<\infty\), and defining
\[
g(Z)
:=
\mathbb E_{P_0}
\left[
\exp\left(\frac{Z}{S}\right)-1
\right].
\]
Since \(S\) is an order unit, \(Z/S\) is bounded for every
\(Z\in\Lip\), so \(g\) is finite-valued, convex, and increasing. Let
\(\phi(q):=q\log q-q+1\) for \(q\geq0\), with \(0\log0:=0\). The
standard integral conjugacy formula for the exponential function gives
\[
g(Z)
=
\sup_{\substack{q:\Omega\to[0,\infty)\ \mathrm{Borel}\\
\mathbb E_{P_0}[\phi(q)]<\infty}}
\left\{
\mathbb E_{P_0}\!\left[q\frac{Z}{S}\right]
-\mathbb E_{P_0}[\phi(q)]
\right\}.
\]
For every such \(q\), Young's inequality yields
\(q/S\leq\phi(q)+e^{1/S}-1\) and
\(q\leq\phi(q)+e-1\). Hence
\(\mathbb E_{P_0}[q/S]<\infty\) and
\(\mathbb E_{P_0}[q]<\infty\). Define
\(d\mu_q=(q/S)\,dP_0\) and
\(\eta_q:=\mu_q-\mu_q(\Omega)\delta_{\omega_0}\). Since
\(S\geq\kappa d(\cdot,\omega_0)\) for some \(\kappa>0\),
\[
\int_\Omega d(\omega,\omega_0)\,d\mu_q
\leq
\frac{1}{\kappa}\mathbb E_{P_0}[q]
<\infty.
\]
Thus \(\eta_q\in\Signed\), and
\(\mathbb E_{P_0}[qZ/S]=\langle Z,\eta_q\rangle\). Therefore \(g\) is
the supremum of affine
\(\sigma(\Lip,\Signed)\)-continuous functions and is
\(\sigma(\Lip,\Signed)\)-lower semicontinuous.

In this case, \(g(yS)=e^y-1\), so
\(h(y)=e^y-1-y\), which is nonnegative, coercive, and has minimum zero
at \(y=0\). Moreover,
\[
g(\min(1,nS))
=
\mathbb E_{P_0}
\left[
\exp\left(\min\left(\frac1S,n\right)\right)-1
\right]
\leq
\mathbb E_{P_0}
\left[
\exp\left(\frac1S\right)-1
\right].
\]
Therefore, all the preceding assumptions are satisfied. Direct
minimization gives
\[
\rho_g(X)
=
\log
\mathbb E_{P_0}
\left[
\exp\left(-\frac{X}{S}\right)
\right].
\]
Hence, \(\rho_g\) is the entropic risk measure applied to the bounded
benchmark-normalized payoff \(X/S\). Its probability representation is 
\[
\rho_g(X)
=
\sup_{\substack{P\ll P_0\\H(P\mid P_0)<\infty}}
\left\{
-\mathbb E_P\left[\frac{X}{S}\right]
-H(P\mid P_0)
\right\},
\]
where \(H(P\mid P_0)\) denotes relative entropy. The exponential
integrability of \(1/S\) ensures that every \(P\ll P_0\) with finite
relative entropy satisfies \(\mathbb E_P[1/S]<\infty\), and therefore
belongs to \(\mathcal{P}_S\).

\end{example}

We now recover a distributionally robust risk measure defined by a Wasserstein ambiguity set. This is informed by \citet{embrechts2021robust}, for instance.
\begin{example}
\label{ex:wasserstein-robust}
Fix
$P_0\in\mathcal{P}_S\cap\mathcal P_1(\Omega)$
and $\delta>0$. Recall from \cref{thm:convex-dual} that $\mathcal{P}_S$ denotes the class of admissible probability measures associated with $S$.

Define the penalty function $\alpha_\delta^{S}:\mathcal{P}_S\to[0,+\infty]$ by
\[
\alpha_\delta^S(P)
:=
\begin{cases}
0, & P\in\mathcal{P}_S\cap\mathcal P_1(\Omega)
     \text{ and } W_1(P,P_0)\leq\delta,\\
+\infty, & \text{otherwise}.
\end{cases}
\]

The Wasserstein--robust discounted risk functional
$\rho_\delta^{S}:\operatorname{Lip}_0(\Omega)\to\mathbb R$ is then defined by
\[
\rho_\delta^{S}(X)
=
\sup_{P\in\mathcal{P}_S}
\left\{
\mathbb{E}_{P}\!\left[-\frac{X}{S}\right]
-
\alpha_\delta^{S}(P)
\right\}.
\]

This functional is of the canonical dual form established in 
\cref{thm:convex-dual}, with penalty given by the indicator of a Wasserstein--$1$ 
ambiguity set. Equivalently, since $\alpha_\delta^{S}$ is the indicator of 
$\mathcal U_\delta(P_0)$, we may rewrite the risk functional as
\[
\rho_\delta^{S}(X)
=
\sup_{P\in \mathcal U_\delta(P_0)}
\mathbb{E}_{P}\!\left[-\frac{X}{S}\right],
\]
where the ambiguity set
\[ 
\mathcal U_\delta(P_0)
:=
\left\{
P\in\mathcal{P}_S\cap\mathcal P_1(\Omega):
W_1(P,P_0)\leq\delta
\right\}.
\]
is the Wasserstein--$1$ ball of radius $\delta$ centered at $P_0$. Thus 
$\rho_\delta^{S}$ corresponds to a distributionally robust optimization criterion 
in which the decision maker evaluates discounted payoffs $X/S$ under all probability 
measures lying within a prescribed transportation distance of a reference measure 
$P_0$.

The acceptance set of $\rho_\delta^S$ admits a natural characterization in terms of 
the metric geometry of $\Omega$:
\[
A_{\rho_\delta^S}
= \left\{ X\in\operatorname{Lip}_0(\Omega) : 
\mathbb{E}_{P}\!\left[\frac{X}{S}\right] \geq 0 \ 
\text{for all } P\in\mathcal{U}_\delta(P_0) \right\}.
\]

A position $X$ is therefore acceptable if and only if its discounted expected value $\mathbb{E}_{P}[X/S]$ is non-negative under every probability measure within 
Wasserstein--$1$ distance $\delta$ of the reference $P_0$. The radius $\delta$ controls the degree of robustness: as $\delta \to 0$ the ambiguity set collapses to 
$\{P_0\}$ and $A_{\rho_\delta^S}$ reduces to the acceptance set of the single-prior discounted expected loss under $P_0$, while as $\delta\to\infty$ the set $A_{\rho_\delta^S}$ shrinks, reflecting increasingly stringent acceptability requirements as the decision maker entertains ever more distant stress scenarios.

\end{example}

\section*{Declarations}

\noindent\textbf{Funding.}
Marlon Moresco acknowledges financial support from the Fundação de Amparo à Pesquisa do Estado do Rio Grande do Sul (FAPERGS), grant no.~25/2551-0000969-3. Marcelo Brutti Righi acknowledges financial support from the Conselho Nacional de Desenvolvimento Científico e Tecnológico (CNPq), grant no.~302869/2024-7.

\noindent\textbf{Conflict of interest.}
The authors declare that they have no conflict of interest.

\noindent\textbf{Ethics approval.}
Not applicable.

\noindent\textbf{Consent to participate.}
Not applicable.

\noindent\textbf{Consent for publication.}
Not applicable.

\noindent\textbf{Data availability.}
No data were used in this study.

\bibliographystyle{apalike}
\bibliography{survey}

\end{document}

%% file: preamble.tex
\usepackage{amssymb,amsmath}
\usepackage{amsthm}
\usepackage{mathtools}
\usepackage{amsfonts}
\usepackage{MnSymbol}
\usepackage{hyperref}
\usepackage[onehalfspacing]{setspace}
\usepackage{geometry}
\usepackage{cite}
\usepackage[authoryear,round]{natbib}
\usepackage{lipsum}
\usepackage{enumerate}
\usepackage{enumitem}
\usepackage{multirow,array}
\usepackage{chngpage}
\usepackage{caption}
\usepackage{float}
\floatstyle{plaintop}
\restylefloat{table}
\usepackage{graphicx}
\usepackage{subfig}
\usepackage{pgfplots}
\usepackage{tikz}
\usetikzlibrary{
    patterns,
    decorations.pathreplacing,
    backgrounds
}
\usepgfplotslibrary{fillbetween}
\usepackage{pgfplotstable}
\usepackage[misc]{ifsym}
\usepackage[dvipsnames]{xcolor}
\usepackage{textcomp}

\pgfplotsset{compat=1.18}

\hypersetup{
    colorlinks=true,
    urlcolor=green,
    linkcolor=red,
    citecolor=blue
}

\usepackage[
    noabbrev,
    nameinlink,
    capitalize,
    sort
]{cleveref}

\usepackage{todonotes}

\newcommand{\R}{\mathbb{R}}

\DeclareMathOperator*{\Lip}{{Lip}_0}
\DeclareMathOperator*{\Signed}{\mathcal M_0^1}

\crefname{namedtheorem}{Theorem}{Theorems}
\Crefname{namedtheorem}{Theorem}{Theorems}

\theoremstyle{plain}

\newtheorem{theorem}{Theorem}[section]
\newtheorem{lemma}[theorem]{Lemma}
\newtheorem{proposition}[theorem]{Proposition}
\newtheorem{corollary}[theorem]{Corollary}

\newtheorem*{theorem*}{Theorem}
\newtheorem*{lemma*}{Lemma}
\newtheorem*{proposition*}{Proposition}
\newtheorem*{corollary*}{Corollary}
\newtheorem*{assumption*}{Assumption}

\theoremstyle{definition}

\newtheorem{definition}[theorem]{Definition}
\newtheorem{remark}[theorem]{Remark}
\newtheorem{example}[theorem]{Example}

\newtheorem*{definition*}{Definition}
\newtheorem*{remark*}{Remark}
\newtheorem*{example*}{Example}

\numberwithin{equation}{section}